\documentclass[preprint]{elsarticle}
\usepackage{amsthm}
\usepackage{amsmath}
\usepackage{amssymb}
\usepackage{comment}
\usepackage[margin=1.4in]{geometry}

\DeclareMathOperator{\modelA}{\mathfrak{A}}
\DeclareMathOperator{\modelB}{\mathfrak{B}}
\DeclareMathOperator{\modelC}{\mathfrak{C}}

\DeclareMathOperator{\tp}{\mathrm{tp}}
\DeclareMathOperator{\Atom}{\mathrm{Atom}}
\DeclareMathOperator{\cl}{\mathrm{cl}}

\theoremstyle{plain}
\newtheorem{theorem}{Theorem}[section]
\newtheorem{lemma}[theorem]{Lemma}
\newtheorem{corollary}[theorem]{Corollary}

\theoremstyle{definition}
\newtheorem{definition}[theorem]{Definition}

\begin{document}

\title{Complexity of the Ackermann Fragment \\ with One Leading Existential Quantifier}

\author{Reijo Jaakkola}
\ead{reijo.jaakkola@tuni.fi}
\address{Faculty of Information Technology and Communication Sciences, Tampere University, \\ Kalevantie 4, 33100 Tampere, Finland}

\begin{abstract}
We prove that the satisfiability problem of the Ackermann fragment with one leading existential quantifier is \textsc{ExpTime}-complete.
\end{abstract}


\maketitle

\section{Introduction}

After it was realized that the satisfiability problem of first-order logic is undecidable, a major research program emerged with the aim to classify prefix fragments of first-order logic based on their decidability status. This program was successfully completed in the 1980's. We refer the reader to \cite{borger97} for a detailed overview of its main results.

One of the two maximal decidable prefix classes with equality turned out to be the Ackermann fragment $[\exists^*\forall\exists^*]_{=}$ \cite{borger97}, i.e., the fragment of first-order logic composed of all sentences in prenex normal form, with a quantifier prefix that matches the expression $\exists^*\forall\exists^*$. It is well-known that the Ackermann fragment has a \textsc{NExpTime}-complete satisfiability problem \cite{ZEROONELAWSDECISIONPROBLEMS} and a closer analysis of the proof for the lower bound shows that it holds already for the slightly smaller fragment $[\exists^2\forall\exists^*]_{=}$.

It is claimed (without proof) in \cite{borger97} on page 288 that the satisfiability problem of $[\exists^*\forall\exists^*]_=$ can be efficiently reduced to that of $[\forall \exists^*]_=$.\footnote{More specifically, the authors claim that this reduction can be done with the method presented in Exercises 6.2.39 and 6.2.40. However, these exercises assume that the sentences under consideration do not contain equality.} If true, this claim would imply that the complexity of the Ackermann fragment remains \textsc{NExpTime}-hard even without leading existential quantifiers. However, in this paper we show that this claim is false (unless \textsc{NExpTime} = \textsc{ExpTime}). The following is our main result.

\begin{theorem}\label{theorem:main-theorem}
    The satisfiability problem of $[\exists\forall \exists^*]_{=}$ is in \textsc{ExpTime}.
\end{theorem}

In \cite{furer1981alternation} it was proved that already the satisfiability problem of $[\forall \exists^2]$ is \textsc{ExpTime}-hard. Thus we have the following immediate corollary.

\begin{corollary}\label{corollary:main-corollary}
    The satisfiability problem of $[\exists\forall \exists^*]_{=}$ is \textsc{ExpTime}-complete.
\end{corollary}

To prove Theorem \ref{theorem:main-theorem}, we will design an alternating polynomial space procedure which, roughly speaking, starts by guessing the isomorphism type of the leading existential quantifier and proceeds to guess witnesses for universally selected elements.

\section{Preliminaries}

In this paper we will work with vocabularies which do not contain constants and function symbols. We will also assume that there are no relation symbols of arity $0$. We will use the Fraktur capital letters to denote structures, and the corresponding Roman letters to denote their domains.

Let $\tau$ be a vocabulary. An \emph{atomic $\tau$-formula} $\alpha(v_1,\dots,v_m)$ is of the form $R(v_1,\dots,v_m)$, where $R\in \tau$. A $1$-\emph{type} $\pi$ over $\tau$ is a maximally consistent set of unary literals, by which we mean formulas of the form $\alpha(x)$ or $\neg \alpha(x)$, where $\alpha$ is an atomic $\tau$-formula. Given a model $\modelA$ and $a\in A$ we will use $\tp_{\modelA}[a]$ to denote the unique $1$-type which $a$ \emph{realizes} in $\modelA$, namely the set
\[\{\alpha(x) \mid \modelA \models \alpha(a)\} \cup \{\neg \alpha(x) \mid \modelA \not\models \alpha(a)\},\]
where $\alpha$ ranges over atomic $\tau$-formulas.

$1$-type of a single element describes completely the set of quantifier-free formulas it satisfies. For our alternating procedure we will also need a closely related notion which specifies sufficiently large portions of the quantifier-free types of pairs of elements. To define this formally, we need to set up some notation. Consider a quantifier-free formula $\psi$ and suppose that $\{x_1,\dots,x_n\}$ contains all the free variables of $\psi$. We will use $\Atom(\psi)$ to denote the set of subformulas of $\psi$ (note that by definition $\Atom(\psi)$ does not contain equalities between variables). Given a formula $\chi$, we define
\[\sim \chi := \begin{cases}
    \chi' & \text{, if } \chi \text{ is } \neg \chi' \\
    \neg \chi & \text{, otherwise}
\end{cases}\]
\noindent We let $\cl(\psi)$ denote the smallest set which contains $\Atom(\psi)$ and is closed under $\sim$.

Fix two distinct variables $x,y \in \{x_1,\dots,x_n\}$. An $(x,y)$-\emph{substitution} is a mapping $s:\{x_1,\dots,x_n\} \to \{x_1,\dots,x_n\}$ with the property that $s(x) = x$ and $s(y) = y$. Given a $(x,y)$-substitution $s$, we use $\cl(\psi,s)$ to denote the set
\[\{\alpha(s(v_1),\dots,s(v_m)) \mid \alpha(v_1,\dots,v_m) \in \cl(\psi) \text{ and } s(\{v_1,\dots,v_m\}) = \{x,y\}\}.\]
For any $(x,y)$-substitution, a maximally consistent set $\rho \subseteq \cl(\psi,s)$ is called a $(2,\psi)$-\emph{profile}. Note that for a fixed pair $(x,y)$  there are at most $2^{|\psi|}n^n \leq 2^{|\psi|^2}$ $(2,\psi)$-profiles, since $\psi$ has at most $|\psi|$ distinct subformulas.

Consider a model $\modelA$, an assignment $s:\{x_1,\dots,x_n\} \to A$ and two distinct elements $a,b\in \modelA$. Suppose that $s(x) = a$ and $s(y) = b$. We define
\[\tp_{\modelA,s}^\psi[a,b] := \{\alpha(s^*(v_1),\dots,s^*(v_m)) \mid \modelA \models \alpha(s(v_1),\dots,s(v_m)) \text{ and } s(\{v_1,\dots,v_m\}) = \{a,b\}\},\]
where 
\[s^*(v_i) := \begin{cases}
    x & \text{, if } s(v_i) = a \\
    y & \text{, otherwise}
\end{cases}\]
as the $(2,\psi)$-profile that $(a,b)$ \emph{realizes} in $\modelA$. We emphasize that in the above definition $\alpha$ ranges (again) only over the formulas in $\cl(\psi)$.

\section{Proof of the upper bound}

Our goal is to design an alternating procedure running in polynomial space, which determines whether a given sentence $\varphi \in [\exists\forall \exists^*]_{=}$ is satisfiable. Since \textsc{APspace} = \textsc{ExpTime} \cite{alternation}, Theorem \ref{theorem:main-theorem} will follow from this.

Fix a sentence
\[\varphi := \exists z \forall x \exists y_1 \dots \exists y_n \psi(z,x,y_1,\dots,y_n),\]
where $\psi$ is quantifier-free. Throughout this section $\tau$ will denote the set of relation symbols occurring in $\varphi$. For technical convenience we will replace $\psi$ with the following quantifier-free formula
\[\psi(z,z,y_1,\dots,y_n) \land (x = z \lor \psi(z,x,y_1,\dots,y_n))\]
Clearly the resulting sentence is equivalent with $\varphi$. Thus we can assume that $x\neq z$ over models of size two.

We start with an important auxiliary definition. In this definition, and also in the rest of this paper, the free variables of $(2,\psi)$-profiles are $z$ and $x$.

\begin{definition}
    Let $\pi_0$ and $\pi$ be $1$-types over $\tau$ and let $\rho$ be a $(2,\psi)$-profile. A pair $(\modelC,s)$, where $\modelC$ is a model of size $n+2$ and $s:\{z,x,y_1,\dots,y_n\}\to C$ is an assignment for which $s(z) \neq s(x)$, is called a $(\pi,\rho)$-\emph{witness candidate} for $\pi_0$ and $\psi$. A $(\pi,\rho)$-witness candidate $(\modelC,s)$ is called a $(\pi,\rho)$-\emph{witness} for $\pi_0$ and $\psi$, if it satisfies the following requirements.
    \begin{enumerate}
        \item $\tp_{\modelC}[s(z)] = \pi_0$ and $\tp_{\modelC}[s(x)] = \pi$
        \item $\tp_{\modelC,s}^\psi[s(z),s(x)] \cup \rho$ is consistent
        \item $\modelC \models \psi(s(z),s(x),s(y_1),\dots,s(y_n))$
    \end{enumerate}
\end{definition}

Observe that the size of the description of $\modelC$ in a witness $(\modelC,s)$ is in the worst case exponential with respect to $|\varphi|$ (the length of $\varphi$), which causes the description of the whole witness to be too large for our purposes. However, to determine whether
\[\modelC \models \psi(s(z),s(x),s(y_1),\dots,s(y_n))\]
holds, we only need to know, in addition to $C$ and $s$, the truth values that atomic formulas of $\psi$ receive under the assignment $s$ in the model $\modelC$, and these can be described using descriptions which have size polynomial with respect to $|\varphi|$, since $\psi$ has at most $|\psi|$ subformulas. We also note that, besides true atomic formulas, our alternating procedure needs to know the $1$-types of elements of $\modelC$, which can of course also be described with a description of size polynomial with respect to $|\varphi|$.

We will now present the promised alternating procedure. First, we will define an auxiliary 
alternating process \textbf{AckermannSatRoutine} that receives as its input a tuple 
\[(\varphi,\pi_0,c,\pi,\rho),\]
where $c$ is a counter (a natural number), $\pi_0$ and $\pi$ are $1$-types over $\tau$ and $\rho$ is a $(2,\psi)$-profile.

\begin{enumerate}
    \item If $c = 2^{|\tau|}2^{|\psi|^2} + 1$, then \textbf{accept}.
    \item \textbf{Existentially guess} a $(\pi,\rho)$-witness candidate $(\modelC,s)$ for $\pi_0$ and $\psi$.
    \item If $(\modelC,s)$ is not a $(\pi,\rho)$-witness for $\pi_0$ and $\psi$, then \textbf{reject}.
    \item \textbf{Universally choose} an element $i \in C$.
    \item Set $\pi := \tp_{\modelC}[i], \rho := \tp_{\modelC,s}^\psi[f(z),i]$ and $c := c + 1$.
    \item Run \textbf{AckermannSatRoutine}$(\varphi,\pi_0,c,\pi,\rho)$.
\end{enumerate}

\noindent Now we define \textbf{AckermannSat} to be the following procedure.
\begin{enumerate}
    \item If $\varphi$ has a model of size one, then \textbf{accept}.
    \item \textbf{Existentially guess} $1$-types $\pi_0$ and $\pi$.
    \item \textbf{Existentially guess} a $(z,x)$-substitution $s:\{z,x,y_1,\dots,y_n\} \to \{z,x,y_1,\dots,y_n\}$ and a maximally consistent set $\rho \subseteq \cl(\psi,s)$.
    \item Run \textbf{AckermannSatRoutine}$(\varphi,\pi_0,0,\pi,\rho)$.
\end{enumerate}

\noindent Keeping in mind the fact that one can represent $2^{|\tau|}2^{|\psi|^2} + 1$ using only polynomially many bits, it is clear that \textbf{AckermannSat} uses only polynomial amount of space. The following two lemmas establish its correctness.

\begin{lemma}
    If $\varphi$ is satisfiable, then $\mathbf{AckermannSat}$ accepts $\varphi$.
\end{lemma}
\begin{proof}
    If $\varphi$ has a model of size one, then $\mathbf{AckermannSat}$ clearly accepts $\varphi$. On the other hand, if $\modelA$ is a model of $\varphi$ of size at least two, then all the existential guesses can be made in accordance with $\modelA$.
\end{proof}

\begin{lemma}
    If $\mathbf{AckermannSat}$ accepts $\varphi$, then $\varphi$ is satisfiable.
\end{lemma}
\begin{proof}
    Suppose that $\varphi$ does not have a model of size one, but the existential player $\exists$ still has a positional winning strategy $\sigma$ in the alternating reachability game played on the configuration graph of the procedure $\mathbf{AckermannSat}$ on input $\varphi$. Let $\pi_0$ denote the $1$-type that $\sigma$ instructs $\exists$ to choose at the start of the game. Without loss of generality we can assume that the move determined by $\sigma$ in Step 2 depends only on the current $1$-type $\pi$ and the $(2,\psi)$-profile $\rho$ (and not on the value of $c$).
    
    Let $\mathcal{W}$ denote the set of all witnesses that are encountered in those histories of the alternating reachability game where $\exists$ moves according to $\sigma$. Note that strictly speaking \textbf{AckermannSat} does not guess the whole structure $\modelC$ in any of the witnesses $(\modelC,s) \in \mathcal{W}$. In other words there might be tuples of elements of $\modelC$ for which \textbf{AckermannSat} did not specify all of the relation symbols whose interpretations contain them. For definiteness, if \textbf{AckermannSat} did not specify whether $(c_1,\dots,c_k)$ belongs to $R^{\modelC}$, then we specify that it does not belong to $R^{\modelC}$ (in other words we complete the models $\modelC$ in a minimal way, although the exact choice of completion does not matter here).
    
    A pair $(\pi,\rho)$, where $\pi$ is a $1$-type over $\tau$ and $\rho$ is a $(2,\psi)$-profile, is called an \emph{extended} $1$-\emph{type}. Since $2^{|\tau|}2^{|\psi|^2} + 1$ is an upper bound on the number extended $1$-types, for every $(\modelC,s)\in \mathcal{W}$ and $i\in C$ the pair
    \[(\tp_{\modelC}[s(i)],\tp_{\modelC,s}^\psi[s(z),s(i)]),\]
    which we call the extended $1$-type that $i$ \emph{realizes}, is encountered in some history of the alternating reachability game where $\exists$ follows $\sigma$. Let $\Phi$ denote the set of extended $1$-types that are realized by elements in the witnesses that belong to $\mathcal{W}$. 
    
    Our goal is to construct an increasing sequence of models
    \[\modelB_0 \leq \modelB_1 \leq \modelB_2 \leq \dots\]
    with the goal being that their union will be a model of $\varphi$. We start by describing how the model $\modelB_0$ can be constructed. Recall that $\pi_0$ denotes the $1$-type that was selected by $\exists$ at the start of the game. In addition to $\pi_0$, the strategy $\sigma$ instructs $\exists$ to choose an another $1$-type $\pi$ and a $(2,\psi)$-profile $\rho$. We define $\modelB_0$ to be a model consisting of two elements $b_0$ and $b_1$, so that $b_0$ and $b$ will have the $1$-types $\pi_0$ and $\pi$ respectively, and furthermore the pair $(b_0,b)$ realizes the $(2,\psi)$-profile $\rho$.
    
    Before proceeding with the rest of the proof, we first give a high level description of how rest of the models $\modelB_1,\modelB_2, \dots$ will be constructed. Given a model $\modelB_m$ we want to assign witnesses for elements that lack them; such elements are called \emph{defects}. More precisely, a defect is an element $b\in (B_m - \{b_0\})$ for which there does not exists a tuple $(d_1,\dots,d_n) \in B_m^n$ so that 
    \[\modelB_m \models \psi(b_0,b,d_1,\dots,d_n).\]
    These witnesses will be selected in a natural way from $\mathcal{W}$. Of course, to be able to do this, we need to make sure that for every $b\in (B_m - B_{m-1})$ we have that $b$ realizes an extended $1$-type from $\Phi$; this will be guaranteed by construction.
    
    Now, an important technical detail is that we have to be very careful with the way in which we specify the structure of pairs $(b_0,b)$, for $b\neq b_0$. Indeed, \textbf{AckermannSat} will only specify the $(2,\psi)$-profile of $(b_0,b)$ and we can only extend this \emph{after} we have provided a witness for $b$, since the witness structure that we use might enforce additional constraints on the structure of $(b_0,b)$. This technical detail might cause our structures $\modelB_m$ to be incomplete, since at stage $m$ for every $b \in B_m - (B_{m-1} \cup \{b_0\})$ we have only specified what is the $(2,\psi)$-profile of the pair $(b_0,b)$. This will not cause us problems, because we can completely specify the structure of $(b_0,b)$ after we have provided a witness for $b$.
    
    Now we continue with the formal proof. Suppose that we have constructed the model $\modelB_m$ and we want to construct the model $\modelB_{m+1}$. For every defect $b$ we will pick a $(\pi,\rho)$-witness $(\modelC_b,s) \in \mathcal{W}$, where $\pi := \tp_{\modelB_m}[b]$ and $\rho$ is the $(2,\psi)$-profile that $(b_0,b)$ realizes in $\modelB_m$. Without loss of generality we will assume that $s(z) = b_0$ and $s(x) = b$. Furthermore, we will assume that $B_m \cap C_b = \{b_0,b\}$ and for any distinct $b\neq b'$ we have that $C_b \cap C_{b'} = \{b_0\}$. We then define the model $\modelB_{m+1}$ as follows.
    \begin{enumerate}
        \item Its domain is the set
        \[B_m \cup \bigcup_{b \text{ is a defect}} C_b.\]
        \item For every $R \in \tau$ and $(b_1,\dots,b_k) \in B_m^k$, where $k$ is the arity of $R$ and $\{b_1,\dots,b_k\} \neq \{b_0,b\}$, we specify that
        \[(b_1,\dots,b_k) \in R^{\modelB_{m+1}} \Leftrightarrow (b_1,\dots,b_k) \in R^{\modelB_m}.\]
        \item We will specify that $(b_0,b)$ realizes the $(2,\psi)$-profile $\tp_{\modelC_b,s}^\psi[b_0,b] \cup \rho$. Then, for every $R \in \tau$ and $(b_1,\dots,b_k)\in B_m^k$ for which we have not yet specified whether $(b_1,\dots,b_k)$ belongs to $R^{\modelB_{m+1}}$, we specify that $(b_1,\dots,b_k)$ does not belong to $R^{\modelB_{m+1}}$.
        \item For every $c \in (C_b - \{b_0,b\})$, we will specify that $(b_0,c)$ realizes $\tp_{\modelC_b,s}^\psi[b_0,c]$.
        \item For every $R\in \tau$ and $(c_1,\dots,c_k) \in C_b^k$, where $k$ is the arity of $R$ and for every $c\in C_b$ (including $b_0$) $\{c_1,\dots,c_k\} \neq \{b_0,c\}$, we specify that
        \[(c_1,\dots,c_k) \in R^{\modelB_{m+1}} \Leftrightarrow (c_1,\dots,c_k) \in R^{\modelC_b}.\]
    \end{enumerate}
    This completes the construction of $\modelB_{m+1}$.
\end{proof}

It is perhaps worthwhile to try to explain why the above argument fails in the case where there are two (or more) existential quantifiers, since the reason for this seems to be somewhat technical. Suppose that an imaginary variant of \textbf{AckermannStart} would start by guessing, say, the $1$-types of the two elements $a,a'$ that we are going to choose to interpretate the two existentially quantified variables. As the procedure goes through different $1$-types and different ways universally selected elements are related to $a$ and $a'$, the procedure needs to remember how $a$ and $a'$ are related to each other, if it wants to guarantee that the pair $(a,a')$ satisfies quantifier-free formulas in a consistent manner in the different existentially guessed witness structures. But since the procedure can run for an exponential amount of time, it could run out of memory, since describing the full structure of $(a,a')$ will require a description of exponential size.

\section{Conclusions}

In this paper we have fixed an error in the literature around prefix fragments by proving that the complexity of the satisfiability problem of the Ackermann fragment becomes \textsc{ExpTime}-complete, if we allow the sentences to contain at most one leading existential quantifier. To prove the \textsc{ExpTime} upper bound, we designed an alternating polynomial space procedure which tries to essentially construct a model for the input sentence.

There are still open problems concerning the complexity of certain prefix fragments. Perhaps the most challenging one would be to determine the exact complexity of the Ackermann fragment extended with a single unary function. This fragment was proved to be decidable in \cite{Shelah1977DecidabilityOA}, but for a complete proof we refer the reader to the book \cite{borger97}.

\section*{Acknowledgement}

The author would like to thank Antti Kuusisto for his helpful discussions concerning an early version of this paper, and Bartosz Bednarczyk for suggesting several improvements on the presentation.

\bibliographystyle{plainurl}
\bibliography{bibliography}

\end{document}